\documentclass[twoside,leqno]{article}

\usepackage[letterpaper]{geometry}

\usepackage{ltexpprt}
\usepackage{hyperref}
\usepackage{verbatim}

\usepackage[square,numbers]{natbib}
\usepackage{todonotes}
\usepackage{regexpatch}
\makeatletter
\xpatchcmd{\@todo}{\setkeys{todonotes}{#1}}{\setkeys{todonotes}{inline,#1}}{}{}
\makeatother

\usepackage{amsmath}
\usepackage{amsfonts}
\usepackage{mathtools}
\usepackage{thmtools}
\usepackage{hyperref}
\usepackage{natbib}
\usepackage{subfig}
\usepackage{graphicx}
\usepackage{svg}
\usepackage{enumerate}
\usepackage{csquotes}
\usepackage{booktabs}
\usepackage{multirow}

\newtheorem{observation}{Observation}[section]

\newproof{Proofof}{Proof of \autoref{lem:small_valid_submultisets}}
\newproof{Proofsketch}{Proof sketch for Lemma~\ref{lem:small_valid_submultisets}}

\newcommand{\rr}{\mathbb{R}}
\newcommand{\zz}{\mathbb{Z}}
\newcommand{\nn}{\mathbb{N}}
\newcommand{\qq}{\mathbb{Q}}

\newcommand{\norm}[1]{\left\lVert #1 \right\rVert}
\newcommand{\abs}[1]{\left\lvert #1 \right\rvert}

\newcommand{\cone}{cone}

\newcommand{\etal}{et al.}
\newcommand{\ie}{i.\,e.}

\newcommand{\eg}{e.\,g.}

\usepackage{mdframed}
\newmdenv[
  topline=false,
  bottomline=false,
  rightline=false,
  skipabove=\topsep,
  skipbelow=\topsep,
  innertopmargin=0pt,
  innerbottommargin=0pt,
  innerleftmargin=10pt,
  innerrightmargin=0pt,
  middlelinecolor=halfgray,
]{siderules}
\newcommand{\claim}[2]
    {
    \begin{siderules}
    \textbf{Claim:}#1 \smallskip

    \noindent\textit{Proof of the claim:}#2
    \hfill {$\vartriangleleft$}
    \end{siderules}
    }

\date{}

\begin{document}

\newcommand\relatedversion{}

\title{\Large Collapsing the Tower - On the Complexity of Multistage Stochastic IPs\relatedversion}
\author{Kim-Manuel Klein\thanks{University of Kiel. Supported by German Research Foundation (DFG) project KL 3408/1-1.}
\and Janina Reuter\thanks{University of Kiel.}}

\maketitle

\begin{abstract}
    In this paper we study the computational complexity of solving a class of block structured integer programs (IPs) - so called multistage stochastic IPs.
    A multistage stochastic IP is an IP of the form $\max \{ c^T x \mid \mathcal{A} x = b, \,l \leq x \leq u,\, x\text{ integral} \}$ where the constraint matrix $\mathcal{A}$ consists of small block matrices ordered on the diagonal line and for each stage there are larger blocks with few columns connecting the blocks in a tree like fashion.
    Over the last years there was enormous progress in the area of block structured IPs. For many of the known block IP classes - such as $n$-fold, tree-fold, and two-stage stochastic IPs, nearly matching upper and lower bounds are known concerning their computational complexity.
    One of the major gaps that remained however was the parameter dependency in the running time for an algorithm solving multistage stochastic IPs. Previous algorithms require a tower of $t$ exponentials, where $t$ is the number of stages, 
    while only a double exponential lower bound was known.
    In this paper we show that the tower of $t$ exponentials is actually not necessary. We can show an improved running time for the algorithm solving multistage stochastic IPs with a running time of
    $2^{(d\norm{A}_\infty)^{\mathcal{O}(d^{3t+1})}} \cdot poly(d,n)$, 
    where $d$ is the sum of columns in the connecting blocks and $n$ is the number of blocks on the lowest stage. 
    Hence, we obtain the first bound by an elementary function for the running time of an algorithm solving multistage stochastic IPs.
    In contrast to previous works, our algorithm has only a triple exponential dependency on the parameters and only doubly exponential for every constant $t$. By this we come very close the known double exponential bound (based on the exponential time hypothesis) that holds already for two-stage stochastic IPs, i.e. multistage stochastic IPs with only two stages.
    
    The improved running time of the algorithm is based on new bounds for the proximity of multistage stochastic IPs. 
    The idea behind the bound is based on generalization for a structural lemma originally used for two-stage stochastic IPs. While the structural lemma requires iteration to be applied to multistage stochastic IPs, our generalization directly applies to inherent combinatorial properties of multiple stages. Already a special case of our lemma yields an improved bound for the Graver Complexity of multistage stochastic IPs.
\end{abstract}

\section{Introduction}\label{sec:introduction}
We consider (integer) linear programs $P=(x,A,b,c)$ of the form
\begin{equation} \label{IP:main}
\begin{split}
\min c^tx \\
Ax=b \\
x\geq 0
\end{split}
\end{equation}
for a \emph{constraint matrix} $A\in\zz^{m\times n}$ with a specific structure, a \emph{right hand side} vector $b\in\zz^m$, an \emph{optimization goal} $c\in \zz^n$, and a vector $x$ of $n$ variables. The constraint matrix $A$  has non-zero entries in a structure similar to \autoref{fig:multistage}.
%STAGES AND TREE. 
The matrix is structured in blocks of multiple stages with the following properties. Blocks of the same stage have distinct rows and columns and for any lower stage block the set of rows is a subset of the rows of a block in the next higher stage. The subset relation on the rows induces a tree-like structure as indicated by arrows.
\begin{figure}
\centering
    \includegraphics[width=0.5\linewidth]{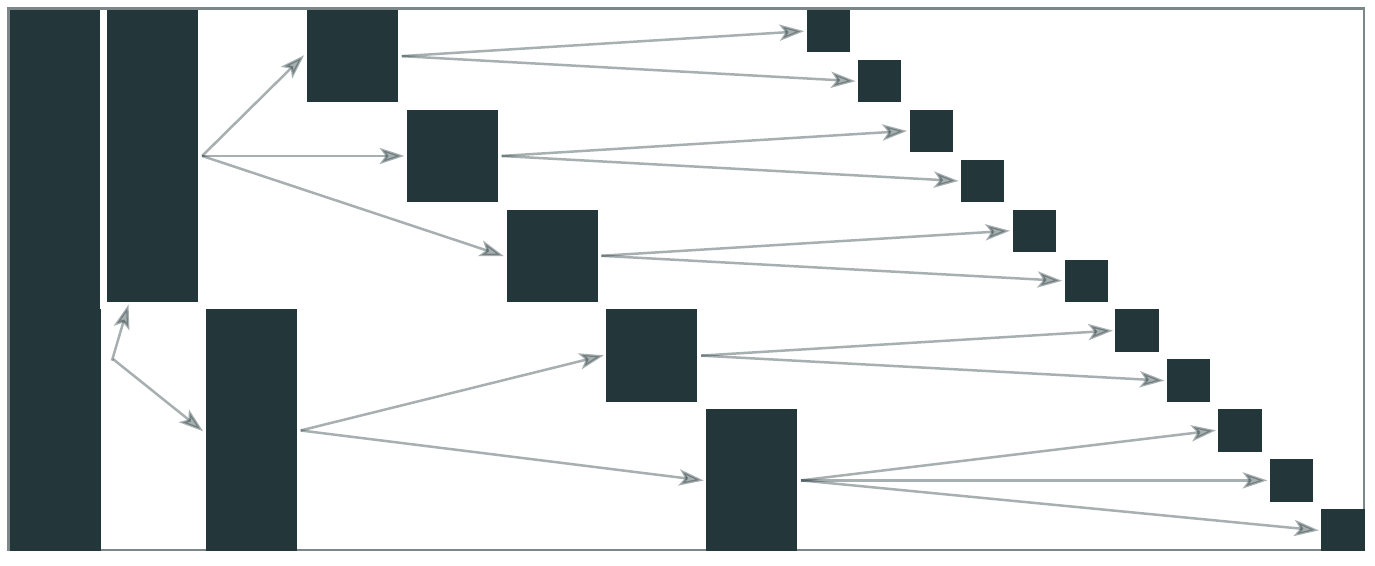}
    \caption{ The structure of nonzero entries in a multistage stochastic matrix is denoted by filled rectangles. Rows of the blocks are connected by a tree indicated by arrows. }
    \label{fig:multistage}
\end{figure}

A famous special case of multistage stochastic IPs are two-stage stochastic IPs, where the constraint matrix $A$ consists only of two stages, a vertical line of block matrices $A^{(i)}\in\zz^{t \times r}$ and a diagonal line of block matrices $B^{(i)}\in \zz^{t \times s}$, i.e.

\begin{align*}
    A = \begin{pmatrix}
A^{(1)} & B^{(1)} & 0 & \cdots & 0\\
A^{(2)} & 0 & B^{(2)} & \ddots &  \vdots \\
\vdots & \vdots & \ddots & \ddots &  0 \\
A^{(n)} & 0 & \cdots & 0 & B^{(n)}
\end{pmatrix}.
\end{align*}

% Matrices appear in practice
Multistage stochastic IPs appear in many real-world problems especially when problems involve uncertainty over time. This particular matrix structure models
\enquote{decisions that occur at different points in time so that the problem can be viewed as having multiple stages of observations and actions}~\cite{birge2011introduction}.
Typically either there are decisions required before all information is revealed~\cite{aircrafts,vehicle_routing} or postponing decisions increases potential costs~\cite{facility, nurses}. A parent block models a decision and its child blocks model all scenarios that might occur in future. The quality of the decision made in a parent block depends on the occurring scenario.
The areas of application include for example worker scheduling~\cite{nurses, scheduling_cross-trained_workers}, project planning~\cite{project_financing, project_planning}, routing problems~\cite{aircrafts, vehicle_routing, droplet_routing}, and facility location planning~\cite{facility}.

% Algorithms for block-ips -> theoretical results for other problems
\subsection{Related Results}
Over the last years there was enormous progress in the development of algorithms solving block IPs, see \autoref{tab:related_work_block_ips} for an overview. 
In the theoretical context there are numerous problems which can be modeled as a block IP and solved more efficiently by algorithms solving the block IP. Faster algorithms for these block structured IPs thus immediately improve the running time for other problems. Applications include string algorithms~\cite{DBLP:conf/esa/KnopKM17}, social choice games~\cite{DBLP:journals/teco/KnopKM20}, scheduling~\cite{DBLP:conf/innovations/JansenKMR19, DBLP:journals/scheduling/KnopK18, DBLP:conf/spaa/JansenLM20}, and bin packing problems~\cite{DBLP:journals/corr/abs-1905-09750}. 

\begin{table}
\centering
\begin{tabular}{ll} 
\toprule
\multicolumn{2}{c}{Overview of results for block IPs}\smallskip\\  
\midrule 
\multicolumn{1}{c}{$n$-fold} & \multicolumn{1}{c}{Two-stage stochastic} \\
\cmidrule(lr){1-1}\cmidrule(lr){2-2}
$\Delta^{\mathcal{O}(t(rs + st))} \cdot n^3\varphi$~\cite{DBLP:journals/mp/HemmeckeOR13} 
& $f_1(r,s,t,\Delta)\cdot poly(n,\varphi)$~\cite{twostagehemmeckeschultz}\\
$(rs\Delta)^{\mathcal{O}(rs^2+sr^2)} \cdot (nt)^2\varphi$~\cite{eisenbrand2018faster} 
& $(rt\Delta)^{(rt\Delta)^{\mathcal{O}(r^2t)}} \cdot (ns)^2\log(ns)\varphi$~\cite{DBLP:conf/ipco/Klein20}\\
 $\Delta^{\mathcal{O}(r^s+sr^2)}\cdot(nt)^3 +\varphi$~\cite{DBLP:conf/icalp/KouteckyLO18} 
& $2^{(2\Delta)^{\mathcal{O}(r^2s+rs^2)}} \cdot n\log^3(n)\varphi$~\cite{AnAlgorithmicTheoryofIntegerProgramming} \\
 $(rs\Delta)^{\mathcal{O}(r^2s+s^2)} \cdot nt\log^{\mathcal{O}(1)}(nt)\varphi^2$~\cite{DBLP:journals/siamdm/JansenLR20}
& $2^{(2\Delta)^{\mathcal{O}(r^2s+rs^2)}} \cdot n^2\log^5n$~\cite{AnAlgorithmicTheoryofIntegerProgramming}\\
$\Delta^{\mathcal{O}(r^2s+rs^2)} \cdot nt\log(nt)\varphi$~\cite{AnAlgorithmicTheoryofIntegerProgramming} 
& $2^{(2\Delta)^{\mathcal{O}(r(r+s))}}\cdot nt\log^{\mathcal{O}(rs)}(nt)$~\cite{DBLP:journals/corr/abs-2012-11742_esa}\\
$\Delta^{\mathcal{O}(r^2s+rs^2)} \cdot (nt)^2\log^3(nt)$~\cite{AnAlgorithmicTheoryofIntegerProgramming} &\\
$(rs\Delta)^{\mathcal{O}(r^2s+rs^2)} \cdot (nt)^{1+o(1)}$~\cite{DBLP:conf/soda/CslovjecsekEHRW21} 
&\smallskip\\
\midrule
\multicolumn{1}{c}{Treefold} & \multicolumn{1}{c}{Multistage stochastic} \\
\cmidrule(lr){1-1}\cmidrule(lr){2-2}
$f_2(t,d,r)\cdot n^3\varphi$~\cite{DBLP:conf/soda/ChenM18} 
& $f_3(d,\Delta,t)\cdot n^3\varphi$~\cite{aschenbrennerhemmeckemultistage}\\
$(d\Delta)^{\mathcal{O}(d^t)} \cdot (rn)^2\log^2(rn)\varphi$~\cite{eisenbrand2018faster} 
& $f_4(d,r,\Delta,t) \cdot (nd)^2\varphi\log^2(nd)$~\cite{DBLP:conf/ipco/Klein20}\\
$(d\Delta)^{\mathcal{O}(d^t)} \cdot (rn)^2\log(rn)\varphi$~\cite{AnAlgorithmicTheoryofIntegerProgramming}
& $f_4(d,r,\Delta, t) \cdot n^2\varphi$\cite{AnAlgorithmicTheoryofIntegerProgramming}\\ 
$(d\Delta)^{\mathcal{O}(d^t)} \cdot (rn)^2\log^3(rn)$~\cite{AnAlgorithmicTheoryofIntegerProgramming}
& $f_4(d,r,\Delta, t) \cdot n^3\log^2n$\cite{AnAlgorithmicTheoryofIntegerProgramming} \\
& $f_5(d,\Delta,t)\cdot rn\log^{\mathcal{O}(2^d)}(rn)$~\cite{DBLP:journals/corr/abs-2012-11742_esa}\\
\bottomrule
\end{tabular}
\caption{In the running times, $\Delta$ denotes the maximum absolute entry of constraint matrix $A$ and $\varphi$ denotes the input length. The functions $f_1$ and $f_3$ are computable and lower bounded by ackerman's function. Function $f_2$ is computable. The functions $f_4$ and $f_5$ involve a tower of exponents of height $t$.}
\label{tab:related_work_block_ips}
\end{table}

% n-fold vs two-stage
Particularly useful for modelling other problems has been a block structure called $n$-fold. This block structure considers the transpose of two-stage stochastic matrices, \ie~the constraint matrix consists of a horizontal line of $n$ blocks $A^{(i)}\in \zz^{r \times t}$ and a diagonal line of $n$ blocks $B^{(i)}\in \zz^{s\times t}$ underneath. Algorithms for $n$-fold IPs are single exponential in the block dimensions, see \eg~\cite{DBLP:conf/soda/CslovjecsekEHRW21}. Though closely related, $n$-fold IPs and two-stage stochastic IPs greatly differ in their complexity. In contrast to the single exponential dependency for $n$-folds, a double exponential lower bound for the dependency on the block dimensions of two-stage stochastic IPs was recently shown under the exponential time hypothesis~\cite{twostagelowerbound}.
The lower bound is complemented by algorithms with double exponential dependency on the block dimensions~\cite{DBLP:conf/ipco/Klein20,AnAlgorithmicTheoryofIntegerProgramming, DBLP:journals/corr/abs-2012-11742_esa}. 

% multistage vs treefold
The transpose of multistage stochastic matrices are treefold matrices. Multistage stochastic IPs and treefold IPs are tree-like structured generalizations of two-stage stochastic IPs and $n$-fold IPs, respectively.
We denote the number of stages by $t$ and the number of blocks on the lowest stage by $n$. In multistage stochastic matrices (treefold matrices) we denote the sum of column (row) dimensions for each stage by $d$ and the number of rows (columns) of the lowest stage by $r$.
Recent work~\cite{eisenbrand2018faster} has shown that for treefold IPs the running time dependency on the block sizes, largest matrix entry, and number of stages behaves similarly than for $n$-fold IPs. It is double exponential but the second exponent only depends on the number of stages and thus the bound is single exponential for any fixed number of stages. 
In contrast, prior results for two-stage and multistage stochastic IPs had a greater gap in their algorithmic complexity. The current best algorithm~\cite{DBLP:journals/corr/abs-2012-11742_esa} for multistage stochastic IPs involves a tower of exponents, where the height of the tower is the number of stages. 
In particular, for any fixed number $k$ of stages the dependency on the block dimensions and largest matrix entry is $k$-fold exponential.

\section{Our Contribution}\label{sec:theorems}
The main result of our paper is to show that multistage stochastic IPs can be solved in time $2^{(d\Delta)^{\mathcal{O}(d^{3t+1})}}\cdot (rn)^{1+o(1)}$. By this we come very close the known double exponential lower bound of $2^{2^{o(d)}}\cdot poly(n)$~\cite{twostagelowerbound} (based on the exponential time hypothesis) that holds already for two-stage stochastic IPs. Our main ingredient to show the improved running time is by a generalization of a structural lemma of Klein which is the key component in \cite{DBLP:conf/ipco/Klein20} for the complexity bound of two-stage stochastic IPs. 

\begin{lemma}[\cite{DBLP:conf/ipco/Klein20}]\label{lem:Kims_lemma}
Given multisets $T_1, \dots, T_n\subset \zz^d_{\geq 0}$ where all elements $\tau\in T_i$ have bounded size $\norm{\tau}_{\infty}\leq \Delta$. Assuming that the total sum of all elements in each set is equal, \ie 
\begin{equation*}
    \sum_{\tau \in T_1}\tau =\,\dots\,= \sum_{\tau \in T_n}\tau
\end{equation*}
then there exist nonempty submultisets $S_1\subseteq T_1, \dots, S_n\subseteq T_n$ of bounded size $\abs{S_i}\leq (d\Delta)^{\mathcal{O}(d\Delta^{d^2})}$ such that 
\begin{equation*}
    \sum_{s \in S_1}s =\,\dots\,= \sum_{s \in S_n}s.
\end{equation*}
\end{lemma}
The structural lemma is also applied in state-of-the-art algorithms for multistage stochastic IPs~\cite{DBLP:conf/ipco/Klein20,DBLP:journals/corr/abs-2012-11742_esa,AnAlgorithmicTheoryofIntegerProgramming}. Simply put, the lemma describes the behavior for one stage in the multistage stochastic IP. Hence, current bounds iterate the lemma over the number of stages to obtain a bound for multistage stochastic IPs. In every iteration the bound grows by one exponent. This iterative application of the lemma seemed rather natural which is why many people in the community believed that a tower of exponents in the running time is actually necessary to solve multistage stochastic IPs.

\begin{figure}
\centering
    \includegraphics[width=0.8\linewidth]{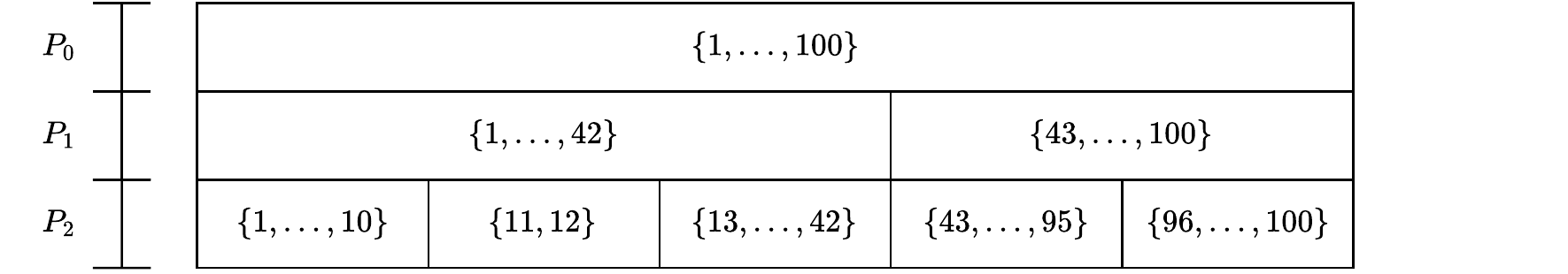}
    \caption{An example for interval partitions with consecutive refinement, where the intervals on every horizontal level are a partition of $\{1, \ldots , 100\}$.}
    \label{fig:partition}
\end{figure}

Our main conceptual result is \autoref{lem:small_valid_submultisets}, a generalization of the structural lemma by Klein. We generalize the lemma to cope directly with arbitrary many stages. We want to state an informal version of the lemma here. For this purpose we define partitions $P_0, \ldots , P_t$ of integral numbers in the interval $[1,n]$  and any partition $P_j$ is a refinement of partition $P_{j-1}$, \ie~for every $i\in P_j$ there exists $i'\in P_{j-1}$ such that $i\subset i'$. For each partition we assign a subset of entries of $\{1,\ldots , d \}$, such that partition $P_i$ is assigned entries $\{s_0+\ldots+s_{i-1}+1, \dots, s_0+\ldots+s_i\}$, with $d = s_0 + \ldots + s_t$. See \autoref{fig:partition} for an example. By $p(\tau,i)$ we denote the projection of vector $\tau \in \zz^d$ to the respective entries. In contrast to the above Lemma, we do not demand equality throughout the entire sum of vectors but only within each interval in the respective vector entries of the partition.

\begin{lemma}\label{lem:main_lemma_intuition}
Given multisets $T_1, \dots, T_n\subset \zz^d_{\geq 0}$ where all elements $\tau\in T_i$ have bounded size $\norm{\tau}_{\infty}\leq \Delta$ and partitions $P_0, \ldots , P_t$ as described above. Assuming that the total sum of all elements in each set is equal in the assigned entries, \ie 
\begin{equation*}
    \sum_{\tau \in T_\ell} p(\tau,i) =\,\dots\,= \sum_{\tau \in T_r}p(\tau,i) \quad \text{ for every interval $I = \{ \ell, \ldots ,r \}$ in partition $P_i$}
\end{equation*}
then there exist submultisets $S_1 \subseteq T_1, \dots, S_n\subseteq T_n$, which are not all empty, of bounded size $\abs{S_i}\leq 2^{(d\Delta)^{\mathcal{O}(d^{3t})}}$ such that 
\begin{equation*}
    \sum_{s \in S_\ell} p(s,i) =\,\dots\,= \sum_{s \in S_r}p(s,i) \quad \text{ for every interval $I = \{ \ell, \ldots ,r \}$ in partition $P_i$.} 
\end{equation*}
\end{lemma}

Using this Lemma, we can show an improved bound for the size of Graver elements of a multistage stochastic matrix $A$. A Graver element is an element of $\ker^{\zz}(A)$ that can not be written as the sum of two non-zero and sign-compatible vectors in the kernel.
\begin{corollary}\label{theorem:graver_multistage}
Let $y$ be a Graver element of multistage stochastic matrix $A$. Then it is bounded by
\[\norm{y}_{\infty} \leq 2^{(d\Delta)^{\mathcal{O}(d^{3t+1})}}.\]
\end{corollary}
Using the algorithm of Eisenbrand et al.~\cite{AnAlgorithmicTheoryofIntegerProgramming} in combination with the improved bound for the size of Graver elements yields an algorithm for solving multistage stochastic IPs with a running time of $2^{(d\Delta)^{\mathcal{O}(d^{3t+1})}} \cdot n^{2} \varphi$, where $\varphi$ is the encoding length of the instance. The IP is of the form~(\ref{IP:main}), where additionally upper and lower bounds on the variables are allowed.

Using the Lemma, we obtain furthermore a statement regarding the proximity of multistage stochastic IPs. An IP has proximity $\rho$ if for every optimal solution $x^*$ to the linear relaxation of the IP there exists an optimal integral solution $x$ such that $\norm{x-x^*}_{\infty}\leq \rho$.
Cslovjecsek \etal~\cite{DBLP:journals/corr/abs-2012-11742_esa} generalized the structural lemma of Klein such that the sum of multisets is allowed to differ slightly in the assumption. Using their generalization, they bounded the proximity of two-stage and multistage stochastic IPs. We show that a similar generalization of our main \autoref{lem:main_lemma_intuition} holds. By this we derive  improved proximity bound for multistage stochastic IPs of the form~(\ref{IP:main}). 
\begin{lemma}\label{lem:proximity}
The proximity of multistage stochastic IPs is bounded by 
$\ \leq 2^{(d\Delta)^{\mathcal{O}(d^{3t+1})}}.$
\end{lemma}

Our proximity bound combined with the algorithmic framework of Cslovjecsek \etal~\cite{DBLP:journals/corr/abs-2012-11742_esa} yields our main theorem. 
\begin{theorem}\label{thm:2-Stage_running_time}
    A multistage stochastic IP of the form (\ref{IP:main}) can be solved in time
    \begin{align*}
        2^{(d\Delta)^{\mathcal{O}(d^{3t+1})}} \cdot (rn)^{1+o(1)}.
    \end{align*}
\end{theorem}

%%%%%%%%%%%%%%%%%%%%%%%%%%%%%%%%%%%%%%%%%%%%%% PRELIMINARIES %%%%%%%%%%%%%%%%%%%%%%%%%%%%%%%
\section{Preliminaries}\label{sec:preliminaries}
For a linear program $P=(x,A,b,c)$ let $Sol^{\rr}(P)$ and $Sol^{\zz}(P)$ denote the sets of fractional and integral solutions, respectively. Denote by $col(A)$ and $row(A)$ the set of column and row indices of a (sub-)matrix $A$, respectively.

\paragraph{Multistage stochastic matrices.}  
We define the shape of the constraint matrix $A$ of a multistage stochastic IP, which we will call a multistage stochastic matrix. The constraint matrix consists of blocks $A^{(1)},\ldots,A^{(\ell)}$ for some $\ell\in\zz_{\geq 0}$, where each block uses a unique set of columns of $A$. The matrix $A$ is multistage stochastic if
\begin{itemize}
    \item there is a block $A^{i_0}$ such that for every $1\leq i\leq \ell$ we have $row(A^{(i)})\subseteq row(A^{(i_0)})$ and
    \item for every two blocks $A^{(i)}, A^{(j)}$ one of the following three conditions $row(A^{(i)})\subseteq row(A^{(j)})$, $row(A^{(i)})\supseteq row(A^{(j)})$, or $row(A^{(i)})\cap row(A^{(j)}) = \emptyset$ is fulfilled.
\end{itemize}

\paragraph{Multistage tree.}
We define a multistage tree $\mathcal{T}(A)=(V,E)$ for any multistage matrix $A$. For every block $A^{(i)}$ there is a vertex $v=col(A^{(i)})\in V$. There is an edge $(u,v)\in E$ if $u$ is the set of columns of $A^{(i)}$, $v$ is the set of columns of $A^{(j)}$, and $row(A^{(i)})\supseteq row(A^{(j)})$. 

This definition is closely related to concepts of primal treedepth of a matrix. If we consider a primal td-decomposition of the primal graph of $A$, then the multistage tree combines vertices on a path in the td-decomposition, where each vertex on the path has exactly one descendent. For more details on this topic we refer to \cite{AnAlgorithmicTheoryofIntegerProgramming}.

\paragraph{Notations.}
The height of the tree is denoted by $t$. We assume that vertices of the same height have same cardinality, \ie~for every $v\in V$ of height $i$ we have $\abs{v}=s_i$. We denote the partial sums of the number of columns by
$d_i := s_0 + \,\dots\, + s_{t-i}$. Note that then $d:=d_0$ matches the primal treedepth $td_P(A)$. 

Let $n$ denote the number of leaves. We assign a number $n(v)\in\{1, \dots, n\}$ to every leaf $v\in V$, where $n(.)$ is a bijective function. 
For every leaf $v$ with corresponding block $A^{(j)}$ and $n(v)=i$ we define a submatrix $A_i$ of $A$. The submatrix $A_i$ consists of the entries $A_{k\ell}$, where $k\in row(A^{(j)})$ and $\ell \in \bigcup_{0\leq j\leq t}v_j$, where $(v_0, \ldots, v_t)$ is a path from the root to leaf $v=v_t$.
Let $P_i=(\tilde{x},A_i,b_i,c_i)$ denote the subprogram of the multistage stochastic IP $P=(x, A, b, c)$, where $b_i$ is the projection of $b$ to the row indices $row(A_i)$ and $c_i$ the projection of $c$ to the column indices $col(A_i)$.

\begin{figure}[b]
\centering
    \includegraphics[width=0.6\linewidth]{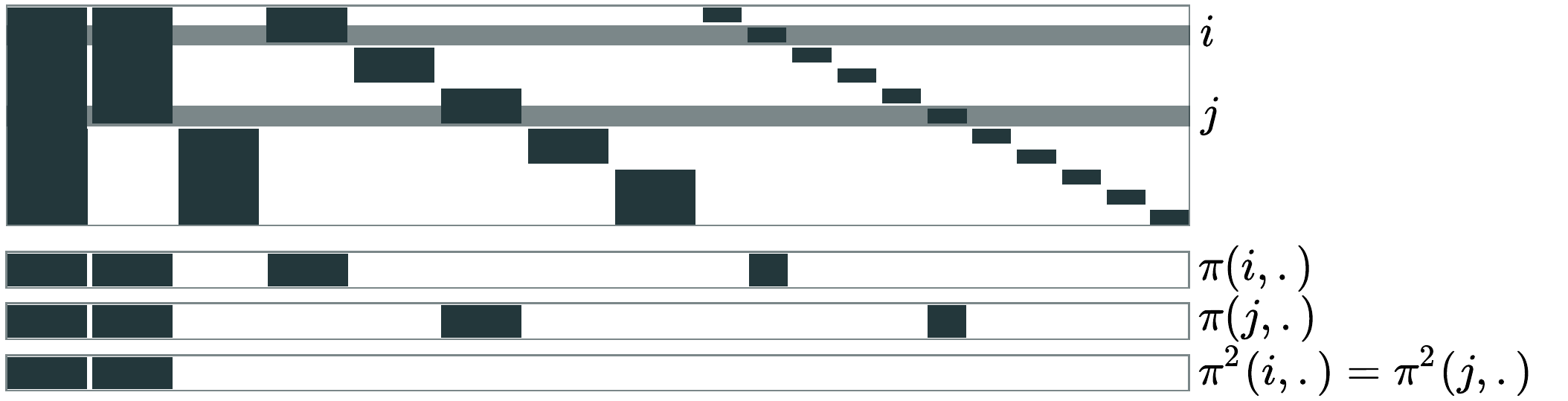}
    \caption{Illustration of the defined projections. The projections depend on the multistage matrix (on top). For $i\leq n$ the function $\pi(i,.)$ projects a vector on the indices of block matrices sharing rows with the leaf $i$, which is marked by filled rectangles (middle part). The variation $\pi^j(i,.)$ again projects the vector removing the right-most $j$ blocks of indices (bottom part).}
    \label{fig:projections}
\end{figure}

Throughout the paper we consider few variants of projecting vectors related to multistage structures. 
Define the function $\pi (i,b)$ for every $i\leq n$ and vector $b\in \rr^{col(A)}$ as the projection of $b$ to the indices $\bigcup_{0\leq j\leq t}v_j$. Note that these vectors are of dimension $d$ as $\abs{v_i} = s_i$.
Let $\pi ^{j}(i,b)$ be the projection of $\pi (i,b)$ to its first $d_j$ indices, which are $\bigcup_{0\leq k \leq t-j}v_k$.
For $b \in \rr^{d_j}$ let $\pi(b)$ be the projection to its first $d_{j+1}$ indices.

\paragraph{Graver bases.}
The conformal (partial) order $\sqsubseteq$ on two vectors $x,y\in \rr^n$ is defined by $x\sqsubseteq y$, if $x_jy_j\geq 0 \text{ and } \abs{x_j}\leq\abs{y_j}$ for all components $j\in \{1, \dots, n\}$. The Graver base $\mathcal{G}(A)$ of an integral matrix $A$ consists of the inclusionwise minimal and non-zero elements of $\ker^{\zz}(A)$.
We are interested in bounding the $\ell_p$ \emph{Graver complexity} of $A$, which is size of 
\begin{equation*}
    g_p(A) := \max_{v\in \mathcal{G}(A)}\norm{v}_p
\end{equation*}
for any $p\in [1,\infty]$.
As remarked by Cslovjecsek \etal~\cite{DBLP:journals/corr/abs-2012-11742_arxiv}, the classical bound for the Graver complexity~\cite{eisenbrand2018faster}, which depends on the number of rows of a matrix, also holds regarding the number of columns of a matrix.
\begin{lemma}[\cite{DBLP:journals/corr/abs-2012-11742_arxiv, eisenbrand2018faster}]\label{lem:graver_size_by_columns}
For every integer matrix $A$ with $m$ columns we obtain the bound
\begin{equation*}
    g_{\infty}(A) \leq (2m\norm{A}_{\infty}+1)^m.
\end{equation*}
\end{lemma}

\paragraph{Proximity.}

Our proof for the proximity of multistage stochastic IPs follows the proof structure of Cslovjecsek \etal~\cite{DBLP:journals/corr/abs-2012-11742_arxiv}. Hence, we use their alternative definition and state two lemmas from their work. Proximity in this sense is a geometric measure that depends only on the polytope $Sol^{\rr}(P)$ and not on the objective function.
\begin{Definition}[\cite{DBLP:journals/corr/abs-2012-11742_arxiv}]
Let $P=(x,A,b,c)$ be a linear program. The \emph{proximity} of $P$, denoted $proximity_{\infty}(P)$,  is the infimum of reals $\rho\geq 0$ satisfying the following: for every fractional solution $x^{frac}\in Sol^{\rr}(P)$ and integral solution $x^{int}\in Sol^{\zz}(P)$ there exists $\tilde{x}^{int}\in Sol^{\zz}(P)$ such that
\begin{equation*}
    \norm{\tilde{x}^{int} - x^{frac}}_{\infty}\leq \rho \quad\text{and}\quad \tilde{x}^{int} - x^{frac} \sqsubseteq x^{int} - x^{frac}.
\end{equation*}
\end{Definition}
Cslovjecsek \etal{} proved that their notion of proximity implies a bound on the usual definition of proximity. Hence, it suffices to prove bounds for $proximity(A)$. Subsequently, they gave a bound for the proximity of arbitrary matrices.
\begin{lemma}[\cite{DBLP:journals/corr/abs-2012-11742_arxiv}]\label{lem:proximity_definitions_are_equiv}
Suppose $P=(x,A,b,c)$ is a linear program. Then for every optimal fractional solution $x^{frac}$ to $P$ there exists an optimal integral solution $x^{int}$ to $P$ satisfying
\begin{equation*}
    \norm{x^{int}- x^{frac}}_{\infty} \leq proximity_{\infty}(P).
\end{equation*}
\end{lemma}
\begin{lemma}[\cite{DBLP:journals/corr/abs-2012-11742_arxiv}]\label{lem:proximity_by_columns}
Let $P=(x,A,b,c)$ be a linear program where $A$ has $m$ columns. Then
\begin{equation*}
    proximity_{\infty}(P)\leq (m\norm{A}_{\infty})^{m+1}.
\end{equation*}
\end{lemma}

\section{On the Structure of Solutions}\label{sec:structure of solutions}
In fact, we prove a more general lemma than \autoref{lem:main_lemma_intuition}, where the sums may differ slightly. We will state this more general version in the context of multistage stochastic IPs. 
Therefore, we require a notion of multisets, where the sums of elements differ slightly in the context of multistage stochastic IPs.

%%%%%%%%%%%%%%%%%%%%%%%%%%%% rho-valid %%%%%%%%%%%%%%%%%%%%%%%%%%%%%%%%%%%%%%%%%
\begin{Definition}
Consider a multistage stochastic matrix $A$ with multistage tree $\mathcal{T}(A)$. Multisets $T_1, \dots, T_n \subset \zz^{d}$ are called  \emph{$\rho$-valid} for $\mathcal{T}(A)$ if there exists $b \in \rr^{col(A)}$ such that for every $i \leq n$ we have
\[\norm{\sum_{\tau\in T_{i}}\tau - \pi(i, b)}_{\infty} < \rho.\]
\end{Definition}
Multisets that are $1$-valid are called \emph{valid}. Valid multisets regarding an integral vector are a tree-like version of equivalence. A special case are two-stage stochastic IPs, where the multistage tree has height $1$. In this case valid multisets are such that the sums of vectors projected to the first $s_0$ indices are equal. Except for the projection, this is the same condition as in \autoref{lem:Kims_lemma}.

Valid multisets are a natural definition for multistage stochastic matrices as they yield a characterization of its integral kernel elements in the following sense. Consider any $b \in\zz^{col(A)}$. If there exist multisets $G_1\subset \mathcal{G}(A_1), \dots, G_n\subset \mathcal{G}(A_n)$ that are valid for $\mathcal{T}(A)$ regarding vector $b$, then we have that 
\begin{equation}\label{eq:valid_multisets}
    \sum_{g \in G_i}g = \pi(i,b)    
\end{equation}
is in the kernel of submatrix $A_i$ for every $i\leq n$ as it is the sum of Graver elements. Hence, $b$ is in the kernel of $A$. 
If otherwise $b\in\ker^{\zz}(A)$, then $\pi(i,b)$ is in the kernel of submatrix $A_i$ for every $i\leq n$ and there exist multisets $G_1\subset \mathcal{G}(A_1), \dots, G_n\subset \mathcal{G}(A_n)$ such that $(\ref{eq:valid_multisets})$ holds for every $i\leq n$. Thus the multisets are valid for $\mathcal{T}(A)$ regarding vector $b$. 
\begin{observation}\label{ob:valid_iff_kernel}
For a vector $b\in\zz^{col(A)}$ there exist multisets $G_1\subset \mathcal{G}(A_1), \dots, G_n\subset \mathcal{G}(A_n)$ that are valid for $\mathcal{T}(A)$ regarding vector $b$ if and only if $b\in\ker^{\zz}(A)$.
\end{observation}

Next we will state the formal version of our main lemma. If $\rho=1$ this is the same statement as \autoref{lem:main_lemma_intuition}. We describe the equivalence briefly. We define partitions $P_0, \dots, P_t$ of the set $\{1, \dots, n\}$ (the numeration of leaves in the multistage tree) as follows. For vertex $v\in V$ denote the set of leaves in the subtree of $v$ by 
\begin{equation*}
    L_v = \{n(v')\mid v'\in V\text{ is a leaf in the subtree of }v\}.
\end{equation*}
For every $0\leq i\leq t$ partition $P_i$ is defined by the sets $L_v$ of vertices $v\in V$ of height $i$. Clearly, this yields a partition for every height as every leaf of the tree is in the subtree of exactly one vertex of that height. \autoref{lem:main_lemma_intuition} requires equality when the sums are projected on the indices of the vertex for the interval. In the following lemma, a vector $b$ of dimension $col(A)$ combines the indices of every vertex in the tree. The sums are compared to the corresponding components of this vector, which is equivalent to the condition in \autoref{lem:main_lemma_intuition}.

%%%%%%%%%%%%%%%%%%%%%%%%%%%% SMALL VALID SUBMULTISETS %%%%%%%%%%%%%%%%%%%%%%%%%%
\begin{lemma}\label{lem:small_valid_submultisets}
Consider a multistage tree $\mathcal{T}(A)$ and multisets $T_1, \dots, T_n \subset \zz^d$, where each $T_i$ contains only sign-compatible elements $\tau$ with  $\norm{\tau}_{\infty}\leq \Delta$. The multisets $T_1, \dots, T_n$ are $\rho$-valid for $\mathcal{T}(A)$ regarding a vector $b\in\rr^{col(A)}$. 

If $\norm{b}_{\infty} > \rho \cdot 2^{(d\Delta)^{\mathcal{O}(d^{3t})}}$, then there exist submultisets $S_i \subseteq T_i$, which are not all empty, and valid for $\mathcal{T}(A)$ with respect to $\hat{b}\in\zz^{col(A)}$ with $\norm{\hat{b}}_{\infty}\leq 2^{(d\Delta)^{\mathcal{O}(d^{3t})}}$.
\end{lemma}
This is also a generalization of Theorem~9 of Cslovjecsek \etal~\cite{DBLP:journals/corr/abs-2012-11742_esa}.
In our notation their bound only applies to multistage trees of height $1$ or in other words to two-stage stochastic matrices. 

The proof of \autoref{lem:small_valid_submultisets} is postponed to \autoref{sec:proof}. Instead, this section focuses on two applications of the lemma. First, we show that the Graver complexity and second we show that the proximity of multistage stochastic matrices is bounded.
Our proof for proximity follows the proof structure of Cslovjecsek \etal~\cite{DBLP:journals/corr/abs-2012-11742_arxiv}.

%%%%%%%%%%%%%%%%%%%%%%%%%%%% GRAVER BOUND %%%%%%%%%%%%%%%%%%%%%%%%%%%%%%%%%%%%%%%
{
\renewcommand{\thecorollary}{\ref{theorem:graver_multistage}}
\begin{corollary}
  The $\ell_{\infty}$ Graver complexity of a multistage stochastic matrix $A$ is bounded by
\[g_{\infty}(A) \leq 2^{(d\Delta)^{\mathcal{O}(d^{3t+1})}}.\]
\end{corollary}
\addtocounter{corollary}{-1}
}
\begin{proof}
Let $y\in\ker^{\zz}(A)$ be a kernel element of $\mathcal{A}$ and assume that $\norm{y}_{\infty} > 2^{(d\Delta)^{\mathcal{O}(d^{3t+1})}}$. By \autoref{ob:valid_iff_kernel} there exist multisets $G_i \subset \mathcal{G}(A_i)$ that are valid for $\mathcal{T}(A)$ regarding $y$ and by \autoref{lem:graver_size_by_columns} the Graver complexity of the submatrices is bounded by $g_{\infty}(A_i) \leq (2d\Delta +1)^{d} =: \Delta'$ for every $1\leq i\leq n$. 

We apply \autoref{lem:small_valid_submultisets} to the multisets $G_i$ and hence there exist submultisets $S_i\subseteq G_i$ for every $i\leq n$ which are valid for $\mathcal{T}(A)$ regarding some $\bar{y} \in \zz^{col(A)}$ with  
\begin{equation*}
    \norm{\bar{y}}_{\infty} \leq 2^{(d\Delta')^{\mathcal{O}(d^{3t})}} \leq 2^{(d(2d\Delta + 1)^{d})^{\mathcal{O}(d^{3t})}} \leq 2^{(d\Delta)^{\mathcal{O}(d^{3t+1})}}.
\end{equation*} 
Using again \autoref{ob:valid_iff_kernel}, but now in the other direction, we get that $\bar{y}\in \ker^{\zz}(A)$. 
At least one submultiset $S_i$ is non-empty. The elements in each set are sign-compatible and non-zero. Hence $\pi(i,\bar{y})$ is nonzero and in particular $\bar{y}$ is. 
The vector $y$ is not in the Graver base of $A$ since it is not minimal by $\bar{y} \sqsubset y$.
\end{proof}

%%%%%%%%%%%%%%%%%%%%%%%%%%%% PROXIMITY BOUND %%%%%%%%%%%%%%%%%%%%%%%%%%%%%%%%%%%%
{
\renewcommand{\thelemma}{\ref{lem:proximity}}
\begin{lemma}
  Suppose $P=(x,A,b,c)$ is a linear program and $A$ is a multistage stochastic matrix. Then
\[proximity_{\infty}(P) \leq 2^{(d\Delta)^{\mathcal{O}(d^{3t+1})}}.\]
\end{lemma}
\addtocounter{lemma}{-1}
}
\begin{proof}
Consider any $x^{frac}\in Sol^{\rr}(P)$ and $x^{int}\in Sol^{\zz}(P)$. Let $\tilde{x}^{int}\in Sol^{\zz}(P)$ be an integral solution such that $\tilde{x}^{int} - x^{frac} \sqsubseteq x^{int} - x^{frac}$ and subject to the condition that $\norm{\tilde{x}^{int} - x^{frac}}_{\infty}$ is minimized.

If there exists a non-zero vector $u\in \ker^{\zz}(A)$ such that $u \sqsubseteq x^{frac} - \tilde{x}^{int}$, then $\tilde{x}^{int}+u \in Sol^{\zz}(P)$ would be a solution with $(\tilde{x}^{int}+u) - x^{frac} \sqsubseteq \tilde{x}^{int} - x^{frac} \sqsubseteq x^{int} - x^{frac}$ and the $\ell_{\infty}$  distance from $x^{frac}$ to $\tilde{x}^{int}+u$ would be strictly smaller than to $\tilde{x}^{int}$. 
The existence of such an $u$ would hence contradict the choice of $\tilde{x}^{int}$. It is sufficient to show that in the case that $\norm{\tilde{x}^{int} - x^{frac}}_{\infty} > 2^{(d\Delta)^{\mathcal{O}(d^{3t+1})}}$ there exists a non-zero vector  $u\in \ker^{\zz}(A)$ such that $u \sqsubseteq x^{frac} - \tilde{x}^{int}$.

For every $i\in\{1, \dots, n\}$ we have that $\pi (i,x^{frac}) \in Sol^{\rr}(P_i)$ and $\pi (i,\tilde{x}^{int}) \in Sol^{\zz}(P_i)$.
Every program $P_1, \dots, P_n$ has $d$ columns. By \autoref{lem:proximity_by_columns}, the proximity is bounded by $proximity_{\infty}(P_i)\leq (d \Delta)^{d +1} =: \rho$. By definition of proximity there exists an integral solution $\hat{x}^{int}_i \in Sol^{\zz}(P_i)$ such that
\begin{equation*}
    \norm{\hat{x}^{int}_{i} - \pi (i,x^{frac})}_{\infty}\leq \rho \quad 
    \text{ and } \quad 
    \hat{x}^{int}_i - \pi (i,x^{frac}) \sqsubseteq \pi (i,\tilde{x}^{int}) - \pi (i,x^{frac}).
\end{equation*}
We have $\hat{x}^{int}_i - \pi (i,\tilde{x}^{int}) \in \ker^{\zz}(A_i)$ which can be decomposed into a multiset $G_i$ of Graver elements. Then $G_i$ is a multiset of sign-compatible elements of $\mathcal{G}(A_i)$ with
\begin{equation*}
    \hat{x}^{int}_i - \pi (i,\tilde{x}^{int}) = \sum_{g \in G_i}g.
\end{equation*}

We want to apply \autoref{lem:small_valid_submultisets}. Therefore, we show that the multisets $G_i$ are $\rho$-valid for the multistage tree, which is the case since
\begin{equation*}
    \norm{\sum_{g\in G_i}g - \pi (i, x^{frac} - \tilde{x}^{int})}_{\infty} 
    = \norm{\hat{x}^{int}_i - \pi (i,\tilde{x}^{int}) - \pi (i, x^{frac}) + \pi (i,\tilde{x}^{int})}_{\infty}  
    \leq \rho.
\end{equation*}
Let $\gamma := \max_{i\leq n}g_{\infty}(A_i) \leq (2d\Delta + 1)^{d}$ by \autoref{lem:graver_size_by_columns}. 
If  $\norm{x^{frac} - \tilde{x}^{int}}_{\infty} \geq 2^{(d\Delta)^{\mathcal{O}(d^{3t+1})}}$, then \autoref{lem:small_valid_submultisets} can be applied
and there exist non-zero submultisets $S_i\subseteq G_i$ that are valid for $\mathcal{T}(A)$ regarding a vector $u \in \zz^{col(A)}$ with $\norm{u}_{\infty} \leq 2^{(d\Delta)^{\mathcal{O}(d^{3t+1})}}$. By \autoref{ob:valid_iff_kernel} the vector $u \in \ker^{\zz}(A)$ is in the kernel of $A$.

The vector $u$ is non-zero as at least one submultiset is non-empty and every element of the multiset is non-zero and sign-compatible. Further, we have $u\sqsubseteq x^{frac} - \tilde{x}^{int}$ since 
\begin{equation*}
    \tilde{u}_i = \sum_{g \in S_i}g \sqsubseteq \sum_{g \in G_i}g = \hat{x}^{int}_i - \pi (i,x^{frac}) \sqsubseteq \pi (i,\tilde{x}^{int}) - \pi (i,x^{frac}).
\end{equation*}{\setlength{\medskipamount}{-0.88cm}\medskip}

\phantom{.}\hfill 
\end{proof}

Due to \autoref{lem:proximity_definitions_are_equiv} this bound applies to the proximity of multistage stochastic IPs in the classical sense.

\section{Proof of \autoref{lem:small_valid_submultisets}}\label{sec:proof}
The \emph{cone} and the \emph{convex hull} spanned by vectors $c_1, \ldots, c_k \in \qq^d$ are defined by
\begin{align*}
    &\cone(c_1, \ldots, c_k) = \{\sum_{i=1}^k\lambda_ic_i \mid \lambda\in\rr^k_{\geq 0}\} \quad\text{ and }\\
    &conv(c_1, \ldots, c_k) = \{\sum_{i=1}^k\lambda_ic_i \mid \lambda\in\rr^k_{\geq 0}, \; \norm{\lambda}_{1} = 1\}.
\end{align*}
Roughly speaking, we construct new multisets consisting of elements of the intersection of cones.
The important property of the constructed multisets is, that they \emph{represent} a submultiset for every child in the multistage tree. By transitivity, the multiset for the root then represents a submultiset for every leaf.

Let $P\subset \zz^d$ be the subset of integral vectors bounded by $\Delta$ in infinity norm and let $\mathcal{B}$ denote the set of $d\times d$ bases with integral entries bounded by $\Delta$ in infinity norm. 
As a preparation for the main proof, we elaborate properties for almost partitioning some multisets to one new multiset. 

We represent a multiset $T$ by a \emph{multiplicity vector} $\lambda$, where $\lambda_p$ denotes how often multiset $T$ contains element $p$. 
By allowing fractional multiplicity vectors we may divide a vector $p$ into several parts, e.g. two times \emph{half of a vector $p$}. When a mathematical operation on a multiplicity vector $\lambda$ requires an index $p$ that is not defined for $\lambda$, then it is treated as $\lambda_p = 0$, similar to the fact that this vector is not included in the multiset represented by $\lambda$. 
For a vector $b \in \rr^d$ and a matrix $M\in \rr^{d\times n}$ we write $b\in M$ if $b$ is a column of $M$. 

The following lemma shows the existence of an element in the intersection cone, that can be represented as a fractional submultiset of every multiset $i\leq n$. Such an element $b$ is represented by only vectors of one basis $B$ and by a fractional multiplicity vector $x$ for that basis, in particular $Bx = b$. Every vector $p\in B$ is used $x_p$ times in the representation. As $x$ is uniquely defined by $x = B^{-1}b$, we will treat $B^{-1}b$ as a multiplicity vector for the column vectors of $B$. The proof goes similar to the \emph{stronger Klein bound} in~\cite{DBLP:journals/corr/abs-2012-11742_arxiv}.

%%%%%%%%%%%%%%%%%%%%%%%%%%%% SINGLE ELEMENT %%%%%%%%%%%%%%%%%%%%%%%%%%%%%%%%%%%%%
\begin{lemma}\label{lem:single_element}
Consider fractional multisets represented by multiplicity vectors $\lambda^{(1)}, \dots, \lambda^{(m)}\in \qq^P_{\geq 0}$  and a vector $b\in \rr^d$  such that for every $i \leq m$ we have 
\[\norm{\sum_{p\in P}\lambda^{(i)}_pp - b}_{\infty} < \rho.\]
If $\norm{b}_{\infty} >\rho \cdot (d\Delta)^{\mathcal{O}(d^2)}$, then there exist bases $B^{(1)}, \dots, B^{(m)}\in \mathcal{B}$ and $\hat{b}\in\zz^d$ such that
\begin{enumerate}[(i)]
    \item $0 \leq x^{(i)} \leq \lambda^{(i)}$ for $x^{(i)} := (B^{(i)})^{-1}\hat{b}$,
    \item $\norm{\hat{b}}_{\infty}\leq (d\Delta)^{d^2}$.
\end{enumerate}
\end{lemma}
\begin{proof}
For every $i\leq m$ let $r^{(i)}\in\rr^P$ be such that 
\[\sum_{p\in P}(\lambda^{(i)}+r^{(i)})p = b\]
and w.l.o.g. $\norm{r^{(i)}}_{\infty}\leq \rho$ using only the unit vectors. Define their sum $z^{(i)}:= \lambda^{(i)}+ r^{(i)}$ for every $i\leq m$.
Every $z^{(i)}$ belongs to the polyhedron $Q=\{x \in \rr^P_{\geq 0} \mid  \sum_{p\in P}x_p p = b\}$.
By Minkowski-Weyl theorem~\cite{schrijver1998theory}, the polyhedron can be written as
\begin{equation*}
    Q= conv(u^{(1)}, \dots, u^{(\ell)}) + \cone(v^{(1)}, \dots, v^{(p)}) \quad\text{ for some  }  u^{(1)}, \dots, u^{(\ell)}, v^{(1)}, \dots, v^{(p)} \in \zz^P
\end{equation*}
where $u^{(j)}\geq \mathbf{0}$ and $ \sum_{p\in P}u^{(j)}_p p = b$, and $v^{(j)}\geq \mathbf{0}$ and $ \sum_{p\in P}v^{(i)}_p p = \mathbf{0}$.

Every $u^{(j)}$ is a vertex solution to the linear program of $Q$ and has at most $d$ non-zero entries and an invertible matrix $C^{(j)}$ such that $C^{(j)}u^{(j)} = b$. 
Every $z^{(i)}$ can be written as 
\begin{equation*}
    z^{(i)} = \sum_{j=1}^{\ell}\gamma_ju^{(j)} + \sum_{k=1}^p\mu_kv^{(k)}, \quad\text{ where }\sum_{j=1}^{\ell}\gamma_j=1\text{ and }\gamma_j, \mu_k\in\rr_{\geq 0}.
\end{equation*}
Hence, for every $i\leq m$ there exists by Carathéodory's theorem~\cite{schrijver1998theory} and by pigeonhole principle an index $j(i)\leq \ell$ with $\gamma_{j(i)}\geq \frac{1}{d+1}$. Since all scalars and vectors are non-negative, we have $0\leq \frac{1}{d+1}u^{(j(i))}\leq z^{(i)}$. 

Consider the intersection $\bigcap_{i=1}^{\ell}\cone(C^{(i)})$, which is a cone with some generating set $C\subset \zz^d$. A consequence of the Farkas-Minkowski-Weyl theorem, see e.g. \cite{schrijver1998theory}, is that the generating elements can be bounded by $\norm{c}_{\infty}\leq (d\Delta)^{d^2}$ for every $c\in C$ as described in \cite{DBLP:journals/corr/abs-2012-11742_arxiv}.

For every $i\leq \ell$ we have that $b \in \cone(C^{(i)})$. Hence, the vector is also in the intersection $b\in \cone(C)$ and there exist by Carathéodory's theorem $d$ vectors $c_1, \dots, c_d\in C$ with 
\begin{equation*}
    b= \sum_{k=1}^{d} \alpha_kc_k \quad \text{for some } \alpha\in \rr^{d}_{\geq 0}.
\end{equation*}
By the assumption on the size of $b$ we have that
\begin{equation*}
    \norm{b}_{\infty} > (d+1) d \cdot 2\rho \cdot \max_{i}\norm{c_i}_{\infty} = \rho \cdot (d\Delta)^{\mathcal{O}(d^2)}.
\end{equation*}
By pigeonhole principle there exists $\alpha_k > 2\rho$ and without loss of generality assume $k=1$. For each $j\in \{1, \dots, \ell \}$ there exist $y^{(j)}, \tilde{y}^{(j)}\in \rr^d_{\geq 0}$ with $C^{(j)} y^{(j)} = c_1$ and $C^{(j)} \tilde{y}^{(j)} = \frac{b}{(d+1)} - 2\rho c_1$. We can write 
\[ C^{(j)} u^{(j)} / (d+1) = b / (d+1) = C^{(j)}(2\rho y^{(j)} + \tilde{y}^{(j)}).\]
We have that $2\rho y^{(j)} \leq u^{(j(i))}/(d+1)$ since $C^{(j)}$ is invertible and $\tilde{y}^{(j)}$ is non-negative.

We set $\hat{b} := c_1$ and $B^{(i)} := C^{(j(i))}$. 
The size of the vector is bounded by $\norm{\hat{b}}_{\infty} = \norm{c_1}_{\infty} \leq (d\Delta)^{d^2}$, which is property $(ii)$. 
For the submultiset relation $(i)$ recall that 
\begin{equation*}
0 \leq 2\rho y^{(j(i))} \leq \frac{1}{d+1}u^{(j(i))} \leq z^{(i)} = \lambda^{(i)}+r^{(i)}.  
\end{equation*}
We have $(B^{(i)})^{-1}\hat{b} =(C^{(j(i))})^{-1}c_1 = y^{(j(i))} \leq \lambda^{(i)}$  since $\norm{r^{(i)}}_{\infty}\leq \rho$.
\end{proof}

The new multiset is then obtained by iterating \autoref{lem:single_element}. Every iteration yields an element, $\hat{b}$ in the above lemma, that is added to the new multiset. Each element represents a submultiset for every $i$ as it is a fractional combination of elements from the original multisets. Property $(i)$ there ensures the subset relation.
Let $P'\subset \zz^d$ be the set of integral vectors bounded by $(d\Delta)^{d^2}$ in infinity norm. 
%%%%%%%%%%%%%%%%%%%%%%%%%%%% NEW MULTISET %%%%%%%%%%%%%%%%%%%%%%%%%%%%%%%%%%%%%%%
\begin{lemma}\label{lem:almost_partitioning}
Consider fractional multisets represented by multiplicity vectors $\lambda^{(1)}, \dots, \lambda^{(m)}\in \qq^P_{\geq 0}$ and a vector $b\in \rr^d$ such that for every $i \leq m$ we have 
\[\norm{\sum_{p\in P}\lambda^{(i)}_pp - b}_{\infty} < \rho.\]
There exist multiplicity vectors $\lambda[B,i]\in\zz^{P'}_{\geq 0}$ for every $B\in \mathcal B$ and $i\leq m$ 
such that 
\begin{enumerate}[(i)]
    \item $\lambda := \sum_{B \in\mathcal{B}} \lambda[B,1] =\, \dots\, = \sum_{B \in\mathcal{B}} \lambda[B,m]$,
    \item $\sum_{p \in P'}\sum_{B \in \mathcal{B}}\lambda[B,i]_p (B^{-1}p) \leq \lambda^{(i)} $ for every $i\leq m$,
    \item if $\lambda[B,i]_p >0$ then $(B^{-1}p) \geq \mathbf{0}$ for every $B\in \mathcal{B}$,  $p\in P$, and $i\leq m$,
    \item $\norm{\sum_{p\in P'}\lambda_p p - b}_{\infty} \leq \rho \cdot (d\Delta)^{\mathcal{O}(d^2)}$.
\end{enumerate}
\end{lemma}
\begin{proof}
We want to iterate \autoref{lem:single_element}. To start the iteration let $\bar{\lambda}^{(i)}[0] = \lambda^{(i)}$  for $i\leq m$ and $\bar{b}[0] := b$. If in iteration $j$ we have that $\norm{\bar{b}[j]}_\infty > \rho \cdot (d\Delta)^{\mathcal{O}(d^2)}$, then we apply \autoref{lem:single_element} to the multisets $\bar{\lambda}^{(i)}[j]$ and there exists bases $B^{(1)}[j], \dots, B^{(m)}[j]\in \mathcal{B}$ and $\hat{b}[j]\in P'$ as stated in the lemma. For the next iteration we define $\bar{b}[j+1] := \bar{b}[j] - \hat{b}[j]$ and $\bar{\lambda}^{(i)}[j+1] := \bar{\lambda}^{(i)}[j] - ((B^{(i)}[j])^{-1}\hat{b}[j])$, where the latter is non-negative by \autoref{lem:single_element}(i). We have that
\begin{equation*}
    \norm{\sum_{p\in P}\bar{\lambda}^{(i)}_p[j+1] p -\bar{b}[j+1]}_{\infty} 
    = \norm{\sum_{p\in P}\bar{\lambda}^{(i)}_p[j] p -\bar{b}[j]}_{\infty} \leq \rho.
\end{equation*}
Let $\nu\in \nn$ be the number of iterations until the condition $\norm{\bar{b}[j]}_\infty > \rho \cdot (d\Delta)^{\mathcal{O}(d^2)}$ does not hold. In particular, the $\ell_{\infty}$-norm of $\bar{b}[\nu]$ is bounded by $\norm{\bar{b}[\nu]}_\infty \leq \rho \cdot (d\Delta)^{\mathcal{O}(d^2)}$.

For every $B\in \mathcal{B}$, $i\leq m$, and $\hat{b}\in P'$ denote by 
\begin{equation*}
    \lambda[B, i]_{\hat{b}}= \abs{\{j\leq \nu \mid \hat{b}[j]=\hat{b}\text{ and }B^{(i)}[j]=B\}}
\end{equation*}
how often basis $B$ was used for multiset $i$ and vector $\hat{b}$. By \autoref{lem:single_element}$(i)$ we have $(B^{-1}p) \geq \mathbf{0}$  if $\lambda[B,i]_p >0$  for every $B\in \mathcal{B}$,  $p\in P$, and $i\leq m$. This is property $(iii)$.
By definition we have for every $\hat{b}\in P'$
\begin{equation*}
    \abs{\{j\leq \nu \mid \hat{b}[j] = \hat b\}} = \sum_{B\in \mathcal{B}}\lambda[B, 1]_{\hat{b}} = \,\dots\, = \sum_{B\in \mathcal{B}}\lambda[B, m]_{\hat{b}}.
\end{equation*}
For every $i\leq m$ we have property $(ii)$ by
\begin{align*}
    \lambda^{(i)} = \sum_{j\leq \nu} (B^{(i)}[j])^{-1}\hat{b}[j] + \bar{\lambda}^{(i)}[\nu]
    \geq\sum_{j\leq \nu} (B^{(i)}[j])^{-1}\hat{b}[j] = \sum_{\hat{b}\in P'}\sum_{B\in\mathcal{B}}\lambda[B,i]_{\hat{b}}\cdot (B^{(i)})^{-1}\hat{b}.
\end{align*}

For property $(iv)$, the sum for the constructed multisets and $b$ differ only by the remaining vector $\bar{b}[\nu]$ since $\sum_{p\in P'}\lambda_pp = \sum_{j\leq \nu}\hat{b}[j]$ and
\begin{align*}
   \norm{\sum_{p\in P'}\lambda_pp - b}_\infty 
   = \norm{\sum_{p\in P'}\lambda_pp - (\sum_{j\leq \nu}b[j] + \bar{b}[\nu])}_\infty 
   = \norm{\bar{b}[\nu]}_\infty \leq \rho\cdot (d\Delta)^{\mathcal{O}(d^2)}.
\end{align*}{\setlength{\medskipamount}{-0.8cm}\medskip}

\phantom{.}\hfill 
\end{proof}

In the proof of \autoref{lem:small_valid_submultisets} we will construct multisets for every vertex in the multistage tree. For vertices of height $0\leq i\leq t$ the elements will be bounded by $\Delta_i := (d\Delta)^{d^{3i}}$ in the $\ell_\infty$ norm. We consider the underlying set
\begin{align*}
    P^j &:= \{ p\in \zz^{d_j} \mid \norm{p}_{\infty} \leq \Delta_j\} 
\end{align*}
and a variant, where elements are of lower dimension
\begin{align*}
    \hat{P}^j &:= \pi(P^j) = \{ p\in \zz^{d_{j+1}} \mid \norm{p}_{\infty} \leq \Delta_j\}.
\end{align*}
As before elements in newly constructed multisets are represented by a certain subset of $\hat{P}^i$ that happens to be a basis. We thus also define the set
\begin{align*}
    \mathcal{B}^{i} &:= \{ B\in \zz^{d_{i+1}\times d_{i+1}} \mid B \text{ is a basis with } \norm{B}_{\infty}\leq \Delta_i \}.
\end{align*}
Let $K_1$ be the constant in the $\mathcal{O}$-notation of \autoref{lem:almost_partitioning}$(iv)$. We will further show that the respective sums for every child differ from the initial vector $b$ only by $\rho_i := \rho \cdot (d\Delta_t)^{iK_1d^{2}}$ for every  height $0\leq i\leq t$.

\begin{Proofsketch}
The proof consists of two phases. The first phase considers the multistage tree in a bottom up fashion and \emph{constructs multisets} for every vertex of the tree using the partitioning lemma,  \autoref{lem:almost_partitioning}. Analyzing \emph{the relation of $\norm{b}_{\infty}$ and $\norm{\pi^t(i,b)}_{\infty}$}, we arrive at either the case that the multiset constructed for the root contains an element of sufficient high multiplicity or, in \emph{the other case}, there exists a vertex, where the multiset contains the vector $\mathbf{0}$ sufficiently often. In both cases the second phase uses the high multiplicity element to \emph{reconstruct submultisets} for every vertex in the subtree in a top down fashion. The reconstruction maintains the important invariant that the sum of elements in the submultisets remains the same for an index once it is defined. One might think of the reconstruction as starting from the indices for $\pi^t(i,b)$ and in each step the vector $b$ is extended by the indices $i\in v'$ for child vertices $v'$.  By this invariant,  the constructed submultisets  are valid for the subtree for every step. At last, the submultisets that are constructed for the leaves of $\mathcal{T}(A)$ are valid for the multistage tree.

We want to preview some technical details. For every vertex $v\in V$ a multiset $\lambda^v$ is considered. 
Another multiset $\hat{\lambda}^v$ considers the projection of the multiset $\lambda^{v}$. In other words $\hat{\lambda}^v$ it is the multiset obtained, when every element of $\lambda^v$ is projected. 
In the second phase, the lemma constructs valid submultisets from an element of $\lambda^v$ with high multiplicity. The reconstruction starting at a vertex $v$ of height $j\leq t$ requires an element $p$ of multiplicity $\lambda^v_p \geq \alpha_j \cdot \beta_j$, for some $\alpha_j, \beta_j\in\nn$ which are defined in the proof below. The reconstruction then extends the vector $\alpha_j \cdot p$, where scaling the vector $p$ with $\alpha_j$ will be used to scale from fractional vectors to integral vectors. In particular, \autoref{lem:almost_partitioning} is used to \emph{fractionally} partition multisets to obtain a multiset for the parent but scaling $p$ with $\alpha_j$ ensures that we find an extension of $\alpha_j \cdot p$ that is an \emph{integral} combination of elements in $T_i$ for every $i\leq n$. On the other side $\beta_j$ leaves room in the multiplicities used for the reconstruction to the multiplicities available in the multiset to use pigeonhole principle when needed.
\end{Proofsketch}

%%%%%%%%%%%%%%%%%%%%%%%%%%%% PROOF OF MAIN LEMMA %%%%%%%%%%%%%%%%%%%%%%%%%%%%%%%%
\begin{Proofof}\phantom{.}
\paragraph{Constructing multisets for the tree.} For all vertices $v\in V$ we construct multisets as follows. If $v$ is a leaf with $n(v) = i$, then let $\lambda^v \in \zz^{P^{0}}_{\geq 0}$ be the multiplicity vector representation of multiset $T_{i}$. This representation is possible since $T_{i} \subset P^{0}$ as every $\tau \in T_{i}$ has dimension $d_0$ and is bounded by $\norm{\tau}_{\infty}\leq \Delta \leq \Delta_0$. Let $\hat{\lambda}^{v}\in \zz^{\hat{P}^{0}}_{\geq 0}$ be defined for every $\hat{p}\in \hat{P}^{0}$ by
\begin{equation*}
    \hat{\lambda}^{v}_{\hat{p}} := \sum_{\substack{p \in P^{0}\\ s.t.\ \pi (p)= \hat{p}}}\lambda^v_p.
\end{equation*}
Note that $\norm{\sum_{p \in P^{0}}\hat{\lambda}^{v}_pp - \pi ^1(i,b)}_{\infty}\leq\norm{\sum_{t\in T_{i}}t - \pi (i,b)}_{\infty}\leq \rho = \rho_0$ since the multisets $T_1, \dots, T_n$ are $\rho$-valid for $\mathcal{T}(A)$.

Consider an inner vertex $v\in V$ of height $j\leq t$, where for all children $v'$ the multisets $\lambda^{v'}\in \zz^{P^{j-1}}$ and  $\hat{\lambda}^{v'} \in \zz^{\hat{P}^{j-1}}$ were already defined and 
\begin{equation*}
    \norm{\sum_{p \in P^{j-1}}\hat{\lambda}^{v'}_pp - \pi ^{j}(i,b)}_{\infty} \leq \rho_{j-1}.
\end{equation*}
We apply \autoref{lem:almost_partitioning} on the multisets $\hat{\lambda}^{v'}$ for every child $v'$ of $v$. By the lemma, there exist multisets $\lambda[B,v']$ for every $v'$ and basis $B \in \mathcal{B}^{j-1}$. Since the elements for multisets $\hat{\lambda}^{v'}$ are in $\hat{P}^{j-1}$, their $\ell_\infty$-norm is bounded by $\Delta_{j-1}$. The infinity norm of elements in the constructed multisets is hence bounded by 
\begin{equation*}
    (d_j\Delta_{j-1})^{d_j^2} 
    = (d_j(d\Delta)^{d^{3(j-1)}})^{d_j^2} 
    \leq (d\Delta)^{d^{3j}} 
    = \Delta_j.
\end{equation*}
Hence each multiset is a subset of $P^{j}$ and can be represented by a multiplicity vector $\lambda[B,v']\in \zz^{P^j}$. Define 
\begin{equation*}
    \lambda^{v} := \sum_{B \in \mathcal{B}^{j-1}}\lambda[B,v']
\end{equation*}
for any child $v'$ of $v$. Note that $\lambda^{v}$ is well-defined as the sum is equal for every child by \autoref{lem:almost_partitioning}$(i)$. 
We define the multiplicity vector $\hat{\lambda}^v\in\zz^{\hat{P}^{j}}_{\geq 0}$ similar to the leaves.
Observe that by \autoref{lem:almost_partitioning}$(iv)$ 
\begin{align*}
    \norm{\sum_{p\in P^j}\lambda^v_pp - \pi ^j(i,b)}_{\infty} 
    \leq \rho_{j-1} \cdot (d_j\Delta_{j-1})^{K_1d_j^2} 
    \leq \rho \cdot (d\Delta_t)^{jK_1d^{2}} 
    = \rho_j.
\end{align*}

\paragraph{The relation of $\norm{b}_{\infty}$ and $\norm{\pi^t(i,b)}_{\infty}$.}
Define $\nu := \text{lcm}(1,\dots, (d\Delta_{t-1})^d)$, $\alpha_i := \nu^i$, and $\beta_i := \Delta_t^{2id^2}$. We will focus on the case that for every $v\in V$ of height $j$ we have $\hat{\lambda}^v_{\mathbf{0}}\leq D_j$ for $D_j := \alpha_j\cdot \beta_j$. The other case, that is $\hat{\lambda}^v_{\mathbf{0}}> D_j$ for some vertex $v\in V$, is discussed at the end of the proof.

Let $v\in V$ be a vertex of height $j$ and consider a child $v'\in V$. Due to $\norm{ \sum_{p\in \hat{P}^{j-1}}\hat{\lambda}^{v'}_pp - \pi ^{j}(i, b) }_{\infty} \leq \rho_{j-1}$ and $\norm{ \sum_{p\in P^{j}}{\lambda}^{v}_pp - \pi ^{j}(i, b) }_{\infty} \leq \rho_{j}$ we have that 
\begin{equation*}
    \norm{\sum_{p\in \hat{P}^{j-1}}\hat{\lambda}^{v'}_pp - \sum_{p\in P^{j}}{\lambda}^{v}_pp}_{\infty} \leq \rho_{j} + \rho_{j-1}.
\end{equation*}
By reverse triangle inequality we get that
\begin{equation*}
    \norm{\sum_{p\in \hat{P}^{j-1}}{\hat{\lambda}}^{v'}_pp}_{\infty} 
    \leq \norm{\sum_{p\in {P}^{j}}{\lambda}^{v}_pp }_{\infty} + \rho_{j} + \rho_{j-1} 
    \leq \Delta_{j}\norm{{\lambda}^{v}}_{1} + \rho_{j} + \rho_{j-1}.
    \stepcounter{equation}\tag{\theequation}\label{eq:lambda_v_and_lambda_v'_1}
\end{equation*}
In order to compare the $\ell_1$-norms of $\lambda^v$ and $\lambda^{v'}$, we consider the left part of the above inequality. By standard arguments we get that
\begin{align*}
    \norm{\sum_{p\in \hat{P}^{j-1}}{\hat{\lambda}}^{v'}_pp}_{\infty} 
    = \norm{\sum_{p\in \hat{P}^{j-1}\setminus \{\mathbf{0}\}}{\hat{\lambda}}^{v'}_pp}_{\infty} 
    \geq \frac{\sum_{p\in \hat{P}^{j-1}\setminus\{0\}}{\hat{\lambda}}^{v'}_p}{d}
    \stackrel{\hat{\lambda}^{v'}_{\mathbf{0}}\leq D_{j-1}}{\geq} \frac{\norm{\hat{\lambda}^{v'}}_{1} - D_{j-1}}{d}.
    \stepcounter{equation}\tag{\theequation}\label{eq:lambda_v_and_lambda_v'_2}
\end{align*}
As a combination of $(\ref{eq:lambda_v_and_lambda_v'_1})$ and $(\ref{eq:lambda_v_and_lambda_v'_2})$ it holds that
\begin{align*}
    \frac{\norm{\hat{\lambda}^{v'}}_{1} - D_{j-1}}{d_{j-1}}
    \leq \Delta_{j}\norm{{\lambda}^{v}}_1 + \rho_{j} + \rho_{j-1} 
    \leq \Delta_{j}(\sum_{p\in P^{j}\setminus\{\mathbf{0}\}}\abs{{\lambda}_p^{v}} + D_{j}) + \rho_{j} + \rho_{j-1}
    \stepcounter{equation}\tag{\theequation}\label{eq:lambda_v_and_lambda_v'}.
\end{align*}
To compare the $\ell_{\infty}$ norms of $\pi^j(i,b)$ and $\pi ^{j-1}(i,b)$, they are compared to the $\ell_1$ norms of $\lambda^v$ and $\lambda^{v'}$, respectively. 
Note that the $\ell_1$ norms of $\lambda^v$ and $\hat{\lambda}^{v}$ are equal by definition.
Using the distance to $b$ we get
\begin{equation*}
    \norm{\hat{\lambda}^{v'}}_1 \geq \frac{\norm{\pi ^{j-1}(i,b)}_{\infty} - \rho_{j-1}}{\Delta_{j-1}} \quad \text{and} \quad
     \sum_{p\in P^{j}\setminus\{0\}}\abs{\lambda_p^{v}} \leq d_{j-1}(\norm{\pi ^{j}(i,b)}_{\infty} + \rho_{j}).
\end{equation*}
Combining the above with $(\ref{eq:lambda_v_and_lambda_v'})$ yields
\begin{align*}
    \norm{\pi ^{j-1}(i,b)}_{\infty} 
    &\leq \Delta_{j}^2d_{j-1}(\norm{\pi ^{j}(i,b)}_{\infty} + 2\rho_{j} + 2\rho_{j-1} + D_{j} + D_{j-1}) 
    \\
    & \leq \Delta_t^2d(\norm{\pi ^{j}(i,b)}_{\infty} + 4\rho_t + 2 D_t).
\end{align*}
For every $i\leq n$ by induction it holds that
\begin{equation*}
    \norm{b}_{\infty} \leq \Delta_t^{2t}d^t_0(\norm{\pi ^{t}(i,b)}_{\infty} + 4t\rho_t + 2t D_t).
\end{equation*}

\paragraph{Reconstructing submultisets.} We remain in the case that $\lambda^v_{\mathbf{0}}\leq D_j$ for every vertex. There exists $p\in P^t\setminus\{0\}$ with
\begin{equation*}
    \lambda^r_{p} \geq \frac{1}{\Delta^d_t}\norm{\lambda^r}_{1} 
    \geq  \frac{1}{\Delta^d_t}\sum_{p\in P^t\setminus\{0\}}\abs{\lambda^t_p} 
    \geq \frac{1}{\Delta^{d+1}_t}(\norm{\pi ^t(i,b)}_{\infty} - \rho_t).
\end{equation*}
If the assumption on the size of $b$
\begin{equation*}
    \norm{b}_{\infty} \geq \Delta^{2t}_td^t(\Delta^{d+1}_t(\alpha_t\beta_t+\rho_t) + 4t\rho_t+2tD_t) = \rho \cdot 2^{(d\Delta)^{\mathcal{O}(d^{3t})}}
\end{equation*}
holds, then $\norm{\pi ^t(i,b)}_{\infty} \geq \Delta^{d+1}_t(\alpha_t\beta_t + \rho_t)$, for all $i\leq n$, and $\lambda^r_{p} \geq \alpha_t\beta_t$. We reconstruct submultisets from the root to every vertex and finally for the leaves. We start with the multiset $\gamma^r := e_p$ for the root.

%%%%%%%%%%%%%%%%%%%%%%%%%%%%%%%%%%%%%%%%%%%%%%%%%%%%%%%%%%%%%%%%%%%%%%%%%%%%
\claim{ Consider a vertex $v\in V$ of height $1 \leq j \leq t$. Then for every $\gamma^v \in\zz^{P^j}_{\geq 0}$ with $\alpha_j\beta_j\cdot \gamma^{v} \leq \lambda^v$ 
there exists $\gamma^{v'}\in \zz^{{P}^{j-1}}_{\geq 0}$ for every child $v'\in V$ such that
\begin{enumerate}[(i)]
    \item $\alpha_{j-1}\beta_{j-1}\cdot \gamma^{v'}\leq \lambda^{v'}$,
    \item $\alpha_{j-1}\norm{\gamma^{v'}}_1\leq d(d\Delta_t)^{d} \alpha_{j} \norm{\gamma^{v}}_1$, and
    \item $\pi (\alpha_{j-1}\sum_{p\in {P}^{j-1}}\gamma^{v'}_pp) = \alpha_{j}\sum_{p\in {P}^{j}}\gamma^v_pp$.
\end{enumerate}  }
{
Consider $\gamma^v \in\zz^{P^j}_{\geq 0}$ with $\alpha_j\beta_j\gamma_p^v \leq \lambda_p^v$ for every $p \in P^j$. In order to prove the claim, first the basis representation of the multiset $\lambda^{v'}$ for every child $v'$ is considered and a basis of sufficient multiplicity $\lambda[B,v']$ is found for each vector in the multiset. Then the vectors are extended to dimension $d_{j-1}$ using a vector with sufficient multiplicity in $\lambda^{v'}$.  Finally, the properties of the claim are verified.

We want to use the representation of the elements in $\lambda^v$ to obtain elements in $\lambda^{v'}$ for any child $v'$. Each representation is defined by the used basis. 
Hence, we start with the basis representation. By definition we have 
\begin{equation*}
    \lambda^{v} = \sum_{B \in \mathcal{B}^{j-1}}\lambda[B,v']
\end{equation*}
for every child $v' \in V$ of $v$. There are at most $\abs{\mathcal{B}^{j-1}}\leq \Delta_{j-1}^{d_{j}^2}$ bases in the set. Hence, for every $p \in P^{j}$ there exists a basis $B^p\in \mathcal{B}^{j-1}$ with $\lambda^v_p / \Delta_{t}^{d^2} \leq \lambda[B^p,v']_p$ by pigeonhole principle. Hence we get 
\begin{align}\label{eq:basispigeonhole}
    \frac{1}{\Delta_t^{d^2}}\alpha_j\beta_j\gamma_p^v 
    \leq \frac{1}{\Delta_t^{d^2}}\lambda_{p}^v 
    \leq \lambda[B^p,v']_{p}.
\end{align}
%%%%%%%%%%%%%%%%%% multiplicity bound
After we selected a basis for every vector, we want to use that the basis and the vector represent a fractional submultiset. To combine the chosen representations in a new multiset, we define a vector $\hat{\gamma}^{v'}\in \zz^{\hat{P}^{j-1}}_{\geq 0}$ by 
\begin{equation*}
    \hat{\gamma}^{v'} :=\nu \cdot \sum_{p \in P^{j}} \gamma^v_{p}\cdot ((B^p)^{-1}p).
\end{equation*} 
The definition is simply the sum of how many times we require which element of the child multiset to represent our current multiset.
Note that $\hat{\gamma}^{v'}_{p'} \in \zz_{\geq 0}$ since all vectors and matrices are integral,  by Cramer's rule $(B^p)^{-1}$ has denominators at most $\abs{\det(B^p)} \leq (d\Delta_{t-1})^d$ which divides $\nu$, and $(B^p)^{-1}p \geq \mathbf{0}$ by \autoref{lem:almost_partitioning}$(iii)$. 

Again, we verify that this representation is in a submultiset relation to the multiset of each child. By \autoref{lem:almost_partitioning}$(ii)$ the following inequality is given
\begin{align}\label{eq:lem43iii}
    \sum_{p \in P^{{j}}}\sum_{\substack{B\in\mathcal{B}^j\\ s.t.\ p'\in B}}\lambda[B,v']_p ((B^p)^{-1}p)_{p'} \leq \hat{\lambda}^{v'}_{p'}.
\end{align}
Using $(\ref{eq:lem43iii})$, the bound in $(\ref{eq:basispigeonhole})$ can be extended to $\hat{\gamma}^{v'}$ and $\hat{\lambda}^{v'}$ as follows
\begin{align*}
            \frac{1}{\Delta_t^{d^2}}\alpha_{j-1}\beta_{j} \cdot \hat{\gamma}_{p'}^{v'}
            =&\frac{1}{\Delta_t^{d^2}}\alpha_{j-1}\beta_{j} \cdot \nu \cdot \sum_{p \in P^{{j}}} \gamma^v_p \cdot  ((B^p)^{-1}p)_{p'} \\
            \stackrel{(\ref{eq:basispigeonhole})}{\leq}& \sum_{p \in P^{{j}}} \lambda[B^{p},v']_p ((B^p)^{-1}p)_{p'} \\
            \leq& \sum_{p \in P^{{j}}}\sum_{B\in\mathcal{B}^j}\lambda[B,v']_p ((B^p)^{-1}p)_{p'}\\
            \stackrel{(\ref{eq:lem43iii})}{\leq}& \hat{\lambda}^{v'}_{p'}. 
            \stepcounter{equation}\tag{\theequation}\label{eq:newmultdefinition}
        \end{align*}
        
%%%%%%%%%%%%%%%%%% backwards projection
Next, each vector is extended to dimension $d_{j-1}$ in order to reverse the projection from $\lambda^{v'}$ to $\hat{\lambda}^{v'}$.
For every $p\in \hat{P}^{j-1}$ there are at most $\Delta_{j-1}^{s_{t-j+1}}\leq \Delta_t^{d}$ vectors $p''\in P^{j-1}$ which are projected to $p'$, \ie{} $\pi (p'')=p'$, and by definition 
\[\hat{\lambda}^{v'}_{p'} = \sum_{\substack{p'' \in P^{j+1}\\ s.t.\ \pi (p'')=p'}}\lambda^{v'}_{p''}.\]
Hence, for every $p' \in \hat{P}^{j-1}$ there exists $p''\in P^{j-1}$ where $\frac{1}{\Delta_t^d}\hat{\lambda}^{v'}_{p'} \leq \lambda^{v'}_{p''}$ and
\begin{align}\label{eq:newgammaineq}
    \frac{1}{\Delta_t^{2d^2}}\alpha_{j-1}\beta_{j} \cdot \hat{\gamma}_{p'}^{v'} \leq \frac{1}{\Delta_t^d}\hat{\lambda}^{v'}_{p'} \leq \lambda^{v'}_{p''}.
\end{align}
Let $f:\hat{P}^{j-1}\mapsto P^{{j-1}}$ map any $p'\in \hat{P}^{j-1}$ to some $p'' \in P^{{j-1}}$ such that $\pi (p'') = p'$ and $p''$ satisfies $(\ref{eq:newgammaineq})$. 
We define the multiset of extended elements $\gamma^{v'}_{p''} \in\zz^{P^{j-1}}_{\geq 0}$ by
\[\gamma^{v'}_{p''} := \begin{cases} \hat{\gamma}^{v'}_{p'} & \text{ if }p''= f(\pi (p'')) \\ 0 & \text{else.}\end{cases}\] 
For the claim it remains to show that $\gamma^{v'}$ satisfies the claimed properties. In particular, for every $p'' \in P^{j+1}$ property $(i)$ holds since
\begin{align}\label{eq:lambdabargamma}
    \alpha_{j-1}\beta_{j-1} \cdot \gamma_{p''}^{v'} \leq \frac{1}{\Delta_t^{2d^2}}\alpha_{j-1}\beta_{j} \cdot \gamma_{p''}^{v'} \leq \lambda^{v'}_{p''}.
\end{align}
%%%%%%%%%%%%% confirm propierty ii
Every $B\in \mathcal{B}^{j-1}$ has $\norm{B}_{\infty}\leq \Delta_{j-1}$. Hence, for every $p\in P^j$ we have by Cramer's rule that
\begin{equation*}
    \norm{(B^p)^{-1}p}_1 \leq d\cdot \norm{(B^p)^{-1}p}_{\infty} \leq d (d\Delta_t)^{d}.
\end{equation*}
By the above we can also bound
\begin{align*}
    \alpha_{j-1} \norm{\hat{\gamma}^{v'}}_1 
    = \alpha_{j} \sum_{p'\in \hat{P}^j} \sum_{p \in P^{{j}}} \gamma^v_{p}( (B^p)^{-1}p)_{p'} 
    \leq d(d\Delta_t)^{d} \alpha_{j} \norm{\gamma^{v}}_1.
\end{align*} 
By the definition of $\hat{\gamma}^{v'}$, we get $\norm{\gamma^{v'}}_1 = \norm{\hat{\gamma}^{v'}}_1$. Combined we can bound the size of the constructed multisets for each child $\alpha_{j-1} \norm{\gamma^{v'}}_1  \leq d(d\Delta_t)^{d} \alpha_{j} \norm{\gamma^{v}}_1$, which is property (ii).

It remains to prove property (iii) $\pi (\alpha_{j-1}\sum_{p\in {P}^{j-1}}\gamma^{v'}_p p) = \alpha_{j}\sum_{p\in {P}^{j}}\gamma^v_p p$, which follows by the definitions of $\hat{\gamma}^{v'}$ and $\gamma^{v'}$.
%
%%%%%%% combination height hat gamma
First, by the definition of $\gamma^{v'}$ the following holds 
\begin{align*} 
     \pi (\alpha_{j-1} \sum_{p'' \in P^{{j-1}}}\gamma^{v'}_{p''} p'') 
     =\alpha_{j-1} \sum_{p'' \in P^{{j-1}}}\gamma^{v'}_{p''} \pi (p'') 
     = \alpha_{j-1}\sum_{p'\in \hat{P}^{j-1}}\hat{\gamma}^{v'}_{p'} p'. \stepcounter{equation}\tag{\theequation}\label{eq:combinationthefirst}
\end{align*}
Second, the definition of $\hat{\gamma}^{v'}$ yields property $(iii)$
\begin{align*}
    \alpha_{j-1}\sum_{p'\in \hat{P}^{j-1}}\hat{\gamma}^{v'}_{p'} p'  
    =\alpha_{j-1}\sum_{p \in P^{{j}}}\sum_{p' \in B^{p}} \gamma^v_{p}(\nu (B^p)^{-1}p)_{p'} p'
    =\alpha_{j}\sum_{p \in P^{{j}}}\gamma^v_{p}B^{p}(B^p)^{-1}p
    =\alpha_j \sum_{p \in P^j}\gamma^{v}_p p.
    \stepcounter{equation}\tag{\theequation}\label{eq:combinationthesecond}
\end{align*}
Hence, combining the equalities (\ref{eq:combinationthefirst}) and (\ref{eq:combinationthesecond}) yields \[\pi (\alpha_{j-1} \sum_{p'' \in P^{{j-1}}}\gamma^{v'}_{p''} p'') = \alpha_j \sum_{p \in P^j}\gamma^{v}_p p.\] 
{\setlength{\medskipamount}{-0.88cm}\medskip}

\phantom{.}\hfill 
}    
% confirm submultisets
We apply the claim iteratively from the root to the leaves on the inner vertices of $\mathcal{T}(A)$. In particular, we construct multisets for every leaf $v\in V$ that satisfy the properties. The constructed multiset of $v$ is a submultiset of $T_{n(v)}$ since by claim property $(i)$ we have that
\begin{equation*}
    \gamma^v = \alpha_0 \beta_0 \cdot \gamma^v \stackrel{(i)}{\leq} \lambda^v
\end{equation*}
and $\lambda^v$ is defined as the multiplicity representation of multiset $T_{n(v)}$.
\paragraph{Verifying the construction.}
Due to claim property $(ii)$ the reconstruction of submultisets grows at most
\begin{equation*}
    \alpha_{j-1} \norm{\gamma^{v'}}_1 \leq d (d\Delta_t)^{d}\alpha_j\norm{\gamma^v}_1
\end{equation*}
from any vertex $v$ to a child $v'$. For a leaf $v$ we have by induction that 
\begin{equation}\label{eq:submultiset_size}
    \norm{\gamma^v}_1 = \alpha_0\norm{\gamma^v}_1 \leq (d(d\Delta_t)^{d})^t \cdot \alpha_t \cdot \norm{\gamma^r}_1 \leq (d\Delta)^{d^{\mathcal{O}(3t)}} \cdot 3^{t(d\Delta_{t-1})^{d}} \cdot 1\leq  2^{(d\Delta)^{\mathcal{O}(d^{3t})}}.
\end{equation}
% Construct combined vector and see that the submultisets are valid

To show that the submultisets are valid for $\mathcal{T}(A)$ we construct a vector $b\in \zz^{col(A)}$. The key observation here is property $(iii)$ from the claim, which ensures that vertices connected in some subtree have equal sums in the indices of the root of the subtree. In each iteration the sum for the parent defines the sum for every child in the respective indices. 
For an index $k\in col(A)$ there exists $v\in V$ of height $j\leq t$ with $k\in v$.
For precise indexing, let $k'$ be the index of $b_k$ in the $d_j$-dimensional vector $\pi^j(i,b)$ for any $i=n(v')$ and $v'$ leaf in the subtree of $v$ .
Define $b_k := \alpha_j \sum_{p\in P^j} \gamma^v_pp_{k'}$. 
Next, we prove that the submultisets are valid for $\mathcal{T}(A)$ regarding $b$.
For a leaf $v\in V$ with $i=n(v)$ consider an index $k\in v'$ for a vertex $v'\in V$ of height $j\leq t$ on the path from the root to $v$. From claim property $(iii)$ used inductively, we get
\begin{equation*}
    \sum_{p \in P^0}\gamma^v_pp_{k'} 
    = \alpha_0 \sum_{p \in P^0}\gamma^v_pp_{k'}
    \stackrel{(iii)}{=}  \alpha_{j} \sum_{p \in P^j}\gamma^{v'}_pp_{k'} 
    \stackrel{\text{def.}}{=} b_k.
\end{equation*}
As this holds for arbitrary indices on the path from the leaf to the root, we have that the constructed submultisets are valid with respect to $b$, \ie
\begin{equation*}
    \sum_{p \in P^0}\gamma^v_p = \pi (i,b)
\end{equation*}
as the vector $\pi (i,b)$ is composed of the indices of $b$ that lie on the path from the root to leaf $v$. Due to the bounded size of the constructed submultisets, the infinity norm of $b$ is bounded by 
\begin{equation*}
    \norm{b}_{\infty} \leq \Delta \cdot \max_{\text{leaf }v}\norm{\gamma^v}_1 \stackrel{(\ref{eq:submultiset_size})}{\leq} \Delta 2^{(d\Delta)^{\mathcal{O}(d^{3t})}} = 2^{(d\Delta)^{\mathcal{O}(d^{3t})}}.
\end{equation*}
% what to do if the 0-assumption does not hold: similar procedure starting with subtree and fill rest with 0's
\paragraph{The other case.} 
Let us now turn to the case that there exists a vertex $v\in V$ with $\hat{\lambda}^v_{\mathbf{0}}> D_j = \alpha_j\beta_j$. 
We proceed very similar to the above but instead of reconstructing from the root, we reconstruct starting from $v$ as we found a vector, that is $\mathbf{0}$, of high multiplicity.
In this case we can apply the claim to the unit vector $\gamma^{v} := e_{\mathbf{0}}$ as it satisfies $\alpha_j\beta_je_{\mathbf{0}} \leq \lambda^v$.
By the claim we construct valid submultisets for the subtree rooted in $v$ of bounded size, similar to the argumentation above. The submultisets for the subtree are non-empty since  $\alpha_j$ extensions of the representation of vector $\mathbf{0}$ are included. Define $b$ as above for indices in the subtree of $v$ and otherwise by $b_i := 0$. By claim property $(iii)$ we have that $\pi ^j(\sum_{p \in P^0}\gamma^{v'}_pp) = \mathbf{0}$, hence if we set for $i$ not in the subtree the multiset $\gamma^v := \mathbf{0}$ as empty, then the submultisets are also valid for $\mathcal{T}(A)$ with respect to $b$.
\end{Proofof}

\begin{comment}

\end{comment}

%\bibliographystyle{plain}
\bibliography{multistage}
\end{document}